\documentclass{article}

\usepackage[preprint]{nips_2018}
\usepackage{graphicx}
\usepackage[toc,page]{appendix}
\usepackage{subcaption}
\usepackage{natbib}

\usepackage[utf8]{inputenc} 
\usepackage[T1]{fontenc}    
\usepackage{hyperref}       
\usepackage{url}            
\usepackage{booktabs}       
\usepackage{amsfonts}       
\usepackage{nicefrac}       
\usepackage{microtype}      
\usepackage{amsmath,amsthm,amssymb}
\usepackage{enumitem}
\newtheorem{theorem}{Theorem}
\newtheorem{prop}{Proposition}  
\title{Versatile Auxiliary Classifier with Generative Adversarial Network (VAC+GAN)}

\author{
Shabab Bazrafkan\\
Dept. Electrical \& Electronic Engineering, College of Engineering \& Informatics\\
National University of Ireland Galway\\
  \texttt{s.bazrafkan1@nuigalway.ie} \\
   \And
Hossein Javidnia\\
Dept. Electrical \& Electronic Engineering, College of Engineering \& Informatics\\
National University of Ireland Galway\\
   \texttt{h.javidnia1@nuigalway.ie} \\
   \AND
Peter Corcoran\\
Dept. Electrical \& Electronic Engineering, College of Engineering \& Informatics\\
National University of Ireland Galway\\
   \texttt{peter.corcoran@nuigalway.ie} \\
}

\begin{document}

\maketitle

\begin{abstract}
One of the most interesting challenges in Artificial Intelligence is to train conditional generators which are able to provide labeled adversarial samples drawn from a specific distribution. In this work, a new framework is presented to train a deep conditional generator by placing a classifier in parallel with the discriminator and back propagate the classification error through the generator network. The method is versatile and is applicable to any variations of Generative Adversarial Network (GAN) implementation, and also gives superior results compared to similar methods.
\end{abstract}

\section{Introduction}
Deep Learning influences almost every aspect of the machine learning and artificial intelligence. It gives superior results for classification, and regression problems compare to classical machine learning approaches \cite{CEmag1}. The other impact of Deep Learning is on generative models \cite{GAN}. In this work, the problem of conditional generators is considered, and a global solution is presented. The conditional generative models are models which can generate a class-specific sample given the right latent input. As one example, these generators can learn the data distribution for male/female faces and produce outputs that match a single (male/female) class. Several researchers have attempted to provide a solution to this problem  \cite{CGAN,ACGAN}. But none of them are able to propose a global solution.\\
In \cite{CGAN} the authors introduce a variation of GAN known as conditional GAN, wherein the model is similar to the ordinary GAN, but the latent space is conditional with respect to the class label. This approach is versatile enough to be extended to other GAN variations, but there is no mathematical proof that the trained generator is able to provide distinct samples for different classes.\\
In our experiments, applying this method to the BEGAN \cite{BEGAN} scheme to generate male/female images did not generate gender-specific samples. This is explained in more detail in section \ref{sec:3}.\\
The most successful implementation of the class specified generative model is Auxiliary Classifier GANs (ACGAN) \cite{ACGAN} wherein by adding a classification term to the generator and discriminator loss, the generator is forced to generate a specific class of data for a given input. see figure \ref{fig:ACGANfig}.\\
\begin{figure}
\centering
\begin{subfigure}[b]{0.49\textwidth}
  \includegraphics[width=\textwidth]{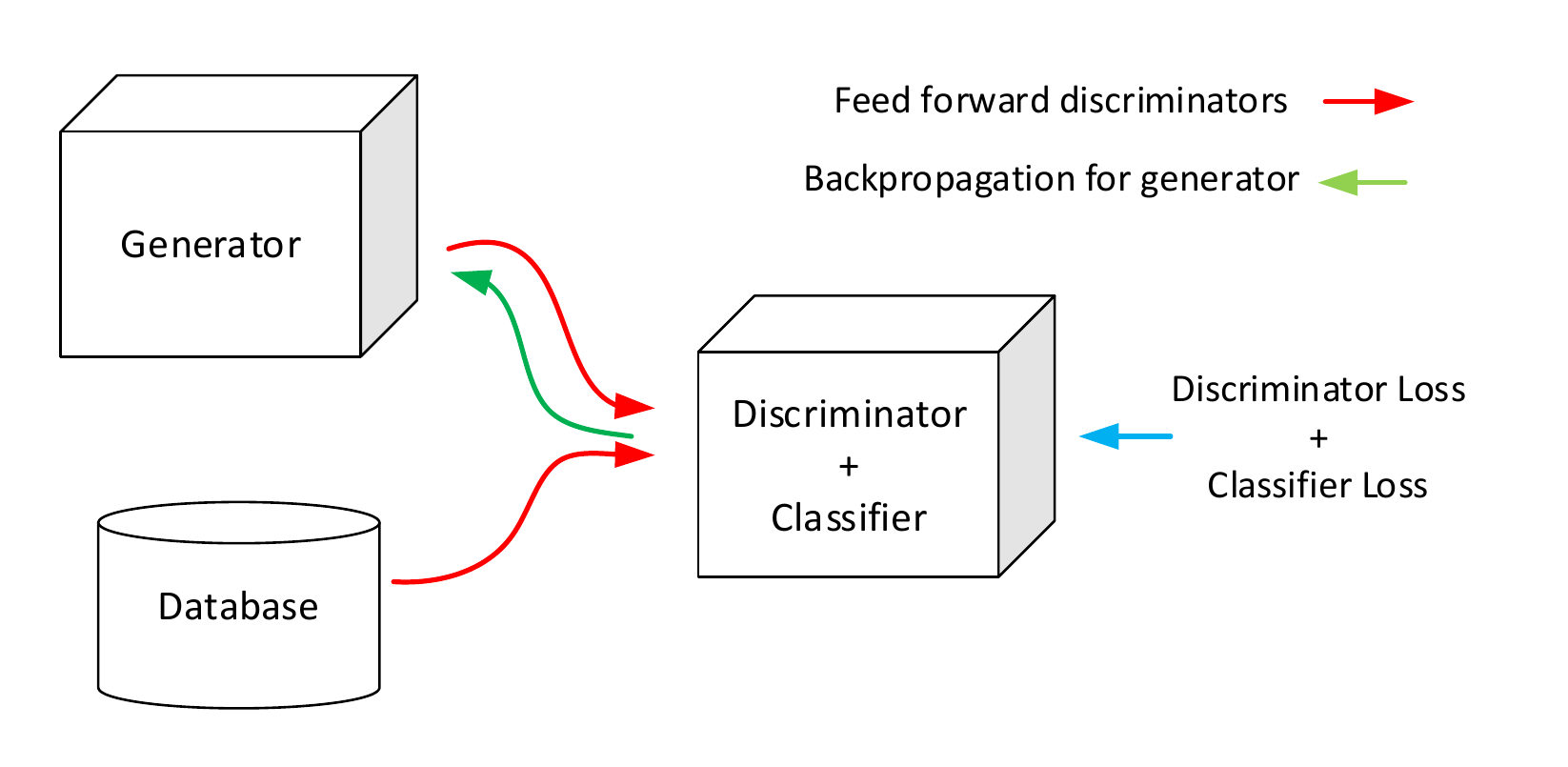}
\caption{The ACGAN scheme.}
\label{fig:ACGANfig}
\end{subfigure}
\begin{subfigure}[b]{0.49\textwidth}
\includegraphics[width=\textwidth]{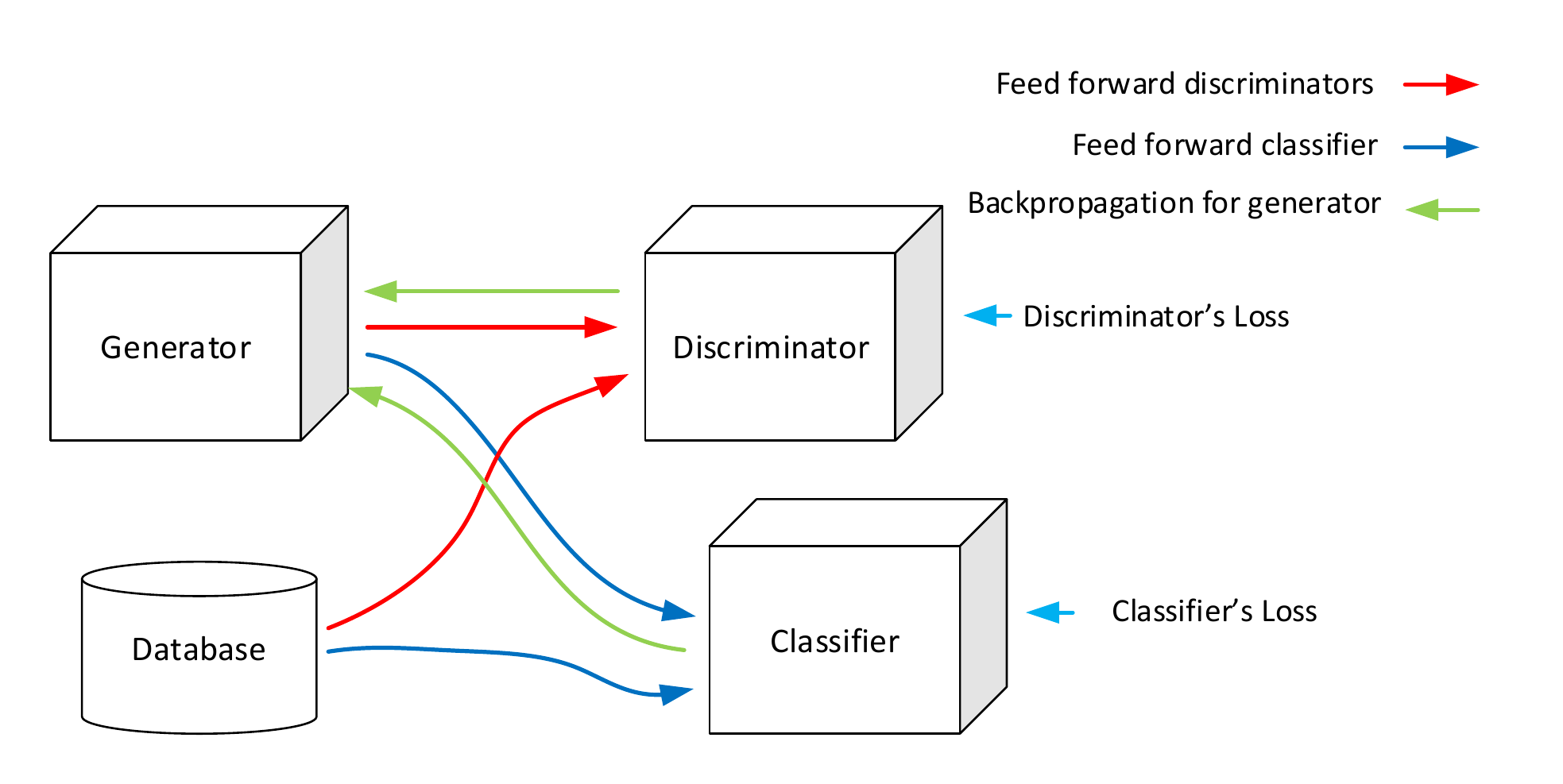}
\caption{The presented scheme (VAC+GAN)}
\label{fig:1}       
\end{subfigure}
\caption{ACGAN vs presented model.}
\end{figure}
The main problem with ACGAN is that it is not versatile enough to be applied to other GAN variations. Mixing the loss of discriminator and the classifier will alter the training convergence specially if the output of the discriminator is from a different type compare to the classifier's output. For example in the BEGAN implementation, the output of the discriminator is an image (2D matrix) compare to the output of the classifier which is a (1D) vector. Merging the loss for two different output types into a single loss alters the convergence of the network.\\
In this work, a new approach for training a class specified generator model is presented which is independent of the generator and discriminator structure. i.e., the presented method can be applied to any model that is already converging. The method is called Versatile Auxiliary Classifier with Generative Adversarial Network (VAC+GAN). The mathematical proof for the effectiveness of the method is also presented. The idea is to get the classification term (in ACGAN) out of the discriminator's loss function by adding a classifier network that back-propagates through the generator.\\
In the next section, the proposed idea is presented alongside with the mathematical proof of the effectiveness of the method, the experimental results are given in section three and conclusions are presented in the last section.

\section{Versatile Auxiliary Classifier with Generative Adversarial Network (VAC+GAN)}
\label{sec:2}
The concept proposed in this research is to place a classifier network in parallel with the Discriminator. The classifier accepts the samples from the generator, and the classification error is back-propagated through the classifier and the generator. The model structure is shown in figure \ref{fig:1}.

In this section, the proposed method is investigated for two class problems. In this case, the classifier is a binary classifier with binary cross-entropy loss function.
The notations used in the mathematical proof are as follows:
\begin{enumerate}
\item $V(G,D)$: is the objective function for a general generative model, wherein $G$ and $D$ are Generator and Discriminator.
\item The latent space $Z$ is partitioned into ${Z_1},{Z_2}$ subsets. This means that $Z_1$ and $Z_2$ are disjoint and their union is equal to the $Z$-space. 
\item $C$ is the classifier function.
\item $\mathcal{L}_{ce}$ is the binary cross-entropy loss function.
\end{enumerate}
\begin{prop}
For a fixed Generator and Discriminator, the optimal Classifier is
\begin{equation}
C_{G,D}^{*}=\frac{p_{X_1}({\bf x})}{p_{X_1}({\bf x})+p_{X_2}({\bf x})}
\end{equation}
wherein $C_{G,D}^{*}$ is the optimal classifier, and $p_{X_1}({\bf x})$, and $p_{X_2}({\bf x})$ are the distributiuons of generated samples for the first and second class respectively.
\end{prop}
\begin{proof}
The objective function for the model is given by:
\begin{equation}
O(G,D,C) = V(G,D)+\mathcal{L}_{ce}(C)
\label{p1e1}
\end{equation}
This can be rewritten as
\begin{equation}
O(G,D,C) = V(G,D)
-\mathbb{E}_{{\bf z}\sim p_{Z_1}({\bf z})}\big[\log(C(G({\bf z})))\big]
-\mathbb{E}_{{\bf z}\sim p_{Z_2}({\bf z})}\big[\log(1-C(G({\bf z})))\big]
\label{p1e2}
\end{equation}
which is given by
\begin{equation}
O(G,D,C) =V(G,D)-\Bigg\{\int p_{Z_1}({\bf z})\log\big(C(G({\bf z}))\big)+p_{Z_2}({\bf z})\log\big(1-C(G({\bf z}))\big)d{\bf z}\Bigg\}
\label{p1e3}
\end{equation}
Considering $G(z_1) = x_1$ and $G(z_2)=x_2$ we get
\begin{equation}
O(G,D,C) = V(G,D)-\Bigg\{\int p_{X_1}({\bf x})\log\big(C({\bf x})\big)
+p_{X_2}({\bf x})\log\big(1-C({\bf x})\big)d{\bf x}\Bigg\}
\label{p1e4}
\end{equation}
The function $f\rightarrow m\log(f)+n\log(1-f)$ reaches its maximum at $\frac{m}{m+n}$ for any $(m,n)\in \mathbb{R}^2 \setminus \{0,0\}$, concluding the proof.
\end{proof}
\begin{theorem}
The maximum value for $\mathcal{L}_{ce}(C)$ is $\log(4)$ and is achieved if and only if $p_{X_1}=p_{X_2}$.
\end{theorem}
\begin{proof}
For $p_{X_1}=p_{X_2} \Longrightarrow C_{G,D}^*=\frac{1}{2}$ and by observing that
\begin{equation}
-\mathcal{L}_{ce}(C) = \mathbb{E}_{{\bf x}\sim p_{X_1}({\bf x})}(\log(C({\bf x})))+\mathbb{E}_{{\bf x}\sim p_{X_2}({\bf x})}(\log(1-C({\bf x})))
\label{p2e1}
\end{equation}
results in 
\begin{equation}
\mathcal{L}_{ce}(C_{G,D}^*)=-\log(\frac{1}{2})-\log(\frac{1}{2})=\log(4)
\label{p2e2}
\end{equation}
To show that this is the maximum value, from equation \ref{p1e4} we have
\begin{equation}
\mathcal{L}_{ce}(C_{G,D}^*)=-\int p_{X_1}({\bf x})\log \Bigg(\frac{p_{X_1}({\bf x})}{p_{X_1}({\bf x})+p_{X_2}({\bf x})} \Bigg) d{\bf x}
-\int p_{X_2}({\bf x})\log \Bigg(\frac{p_{X_2}({\bf x})}{p_{X_1}({\bf x})+p_{X_2}({\bf x})} \Bigg) d{\bf x}
\label{p2e3}
\end{equation}
which is equal to
\begin{equation}
\mathcal{L}_{ce}(C_{G,D}^*)= \log(4)
-\int p_{X_1}({\bf x})\log \Bigg(\frac{p_{X_1}({\bf x})}{\frac{p_{X_1}({\bf x})+p_{X_2}({\bf x})}{2}} \Bigg) d{\bf x}
-\int p_{X_2}({\bf x})\log \Bigg(\frac{p_{X_2}({\bf x})}{\frac{p_{X_1}({\bf x})+p_{X_2}({\bf x})}{2}} \Bigg) d{\bf x}
\label{p2e4}
\end{equation}
results in
\begin{equation}
\begin{split}
\mathcal{L}_{ce}(C_{G,D}^*)= \log(4) -& KL\Bigg(p_{X_1}({\bf x})\Big|\Big| \frac{p_{X_1}({\bf x})+p_{X_2}({\bf x})}{2}\Bigg)\\
-& KL\Bigg(p_{X_2}({\bf x})\Big|\Big| \frac{p_{X_1}({\bf x})+p_{X_2}({\bf x})}{2}\Bigg)
\end{split}
\label{p2e5}
\end{equation}
Where $KL$ is the Kullback-Leibler divergence, which is always positive or equal to zero, concluding the proof. 
\end{proof}
\begin{theorem}
Minimizing the binary cross-entropy loss function $\mathcal{L}_{ce}$ for the classifier $C$ is increasing the Jensen-Shannon divergence between $p_{X_1}$ and $p_{X_2}$.
\end{theorem}
\begin{proof}
the Jensen-Shannon divergence between $p_1$ and $p_2$ is given by
\begin{equation}
JSD(p_1||p_2)=\frac{1}{2}KL\bigg(p_1\big|\big|\frac{p_1+p_2}{2}\bigg)+\frac{1}{2}KL\bigg(p_2\big|\big|\frac{p_1+p_2}{2}\bigg)
\label{p3e1}
\end{equation}
considering equation \ref{p2e5} and \ref{p3e1}, it gives
\begin{equation}
\mathcal{L}_{ce}(C_{G,D}^*)=\log(4)-2JSD(p_{X_1}||p_{X_2})
\end{equation}
minimizing $\mathcal{L}_{ce}$ is equal to maximizing $JSD(p_{X_1}||p_{X_2})$, concluding the proof.
\end{proof}
Here it has been shown that placing the classifier $C$ and add its loss value the generative framework pushes the generator to increase the distance of samples that are drawn from a specific class with respect to the other class. For example, in the male/female face scenario, one can use a partition of $Z$ space to generate male and and another partition to generate female samples.
\section{Experimental Results}
\label{sec:3}
In this section, an experiment is conducted to show the effectiveness of the proposed scheme while different measures are used to show the diversity of the generated samples including Mean Square Error (MSE), Root Mean Square Error (RMSE), Mean Absolute Error (MAE), Universal Quality Index (UQI), and Structural Similarity Index (SSIM). These measurements are explained in Appendix \ref{appendixA}.
MSE, RMSE and MAE show the difference between two images. The higher values for these metrics correspond to higher variation of the generated images. UQI and SSIM measure the structural similarity between two samples. Lower value for these measurements correspond to less similarity. In evaluating generative models higher values for MSE, RMSE, and MAE and lower values for UQI and SSIM is desirable.\\
In this section, all the networks are trained in Lasagne \cite{LASAGNE} on top of Theano \cite{THEANO} library in Python.\\
The experiment is conducted by training a gender specified generator using the BEGAN \cite{BEGAN} structure trained on CelebA database. The results of the proposed method are compared against the results of conditional GAN idea applied to the BEGAN framework. The comparisons with ACGAN method is not available since applying this method to BEGAN framework altered the convergence of the model and the generator constrained to a deterministic output even after the first epoch.\\
The CelebA dataset \cite{CelebA} consisting of 202,599 original images is used for training our GAN framework. The OpenCV frontal face cascade classifier \cite{OpenCV} is used to detect facial regions which are cropped and resized to $48\times 48$ pixels. In the BEGAN framework the generator network is a typical GAN generator which has the same architecture as the decoder part of an auto-encoder. The network used in our experiment contains one fully connected layer which maps the input to a 3D layer. Next layers are all convolutional layers followed by $(2, 2)$ un-pooling layers for every second convolution. The exponential linear unit (ELU) \cite{ELU} is used as activation function except in the last layer wherein no non-linearity has been applied. And the discriminator network is an auto-encoder. The
input of the auto-encoder is the image $(48 \times 48)$. The encoder part of the network is made of convolutional layers with ELU activation function. The downscaling in these layers is obtained by using $(2, 2)$ stride in every second convolutional layer. The architecture of decoder is the same as the generator network. And the bottleneck of the auto-encoder is a fully connected layer with no activation function. The encoder and decoder networks used for training the BEGAN are shown in figures \ref{fig:2} and \ref{fig:3} respectively. The layers shown in red apply no nonlinearity to the data.
\begin{figure}
\begin{subfigure}[b]{.49\textwidth}
  \includegraphics[width=\columnwidth]{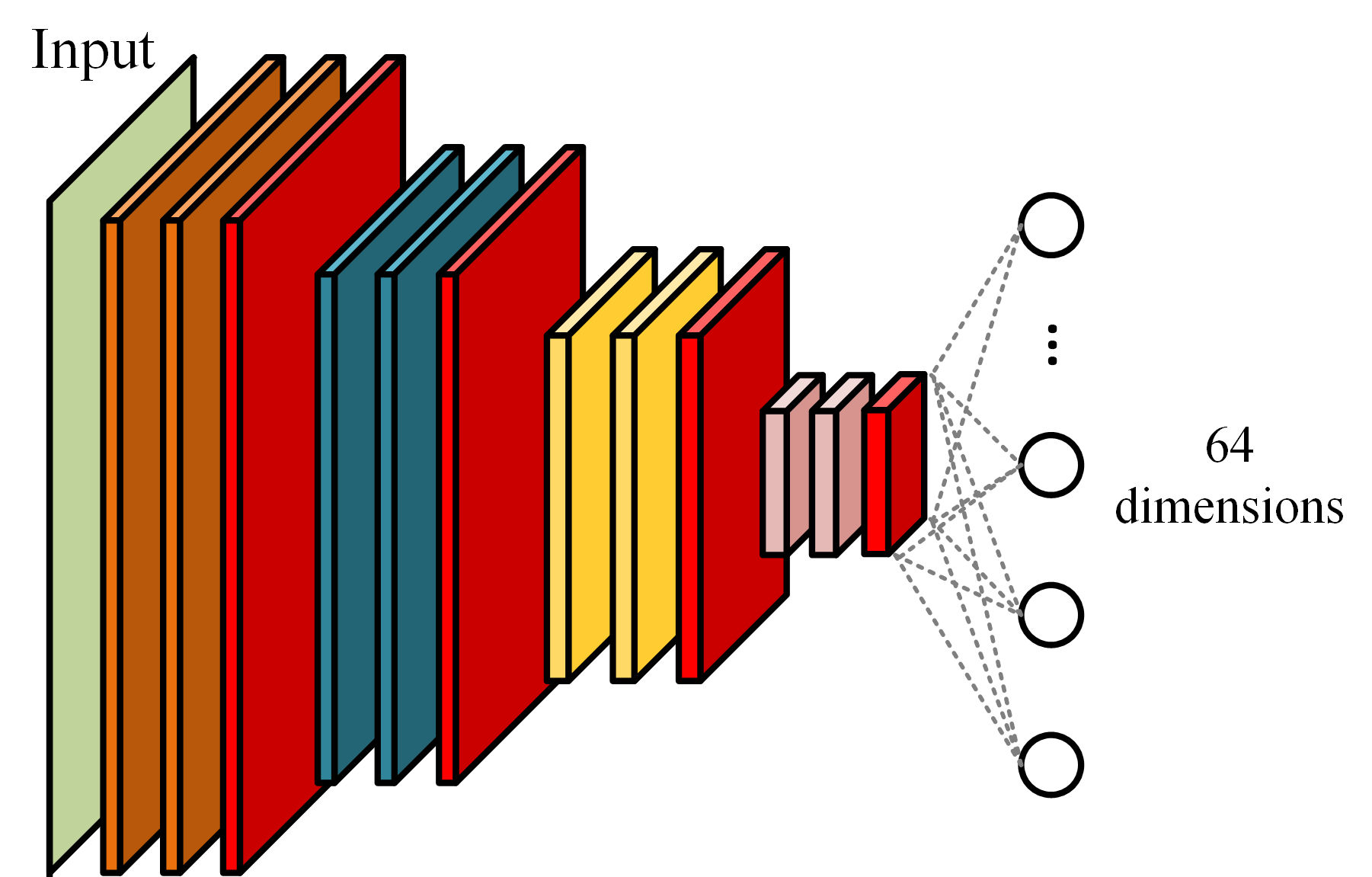}
\caption{Encoder.}
\label{fig:2}   
\end{subfigure} 
\begin{subfigure}[b]{.49\textwidth}
  \includegraphics[width=\columnwidth]{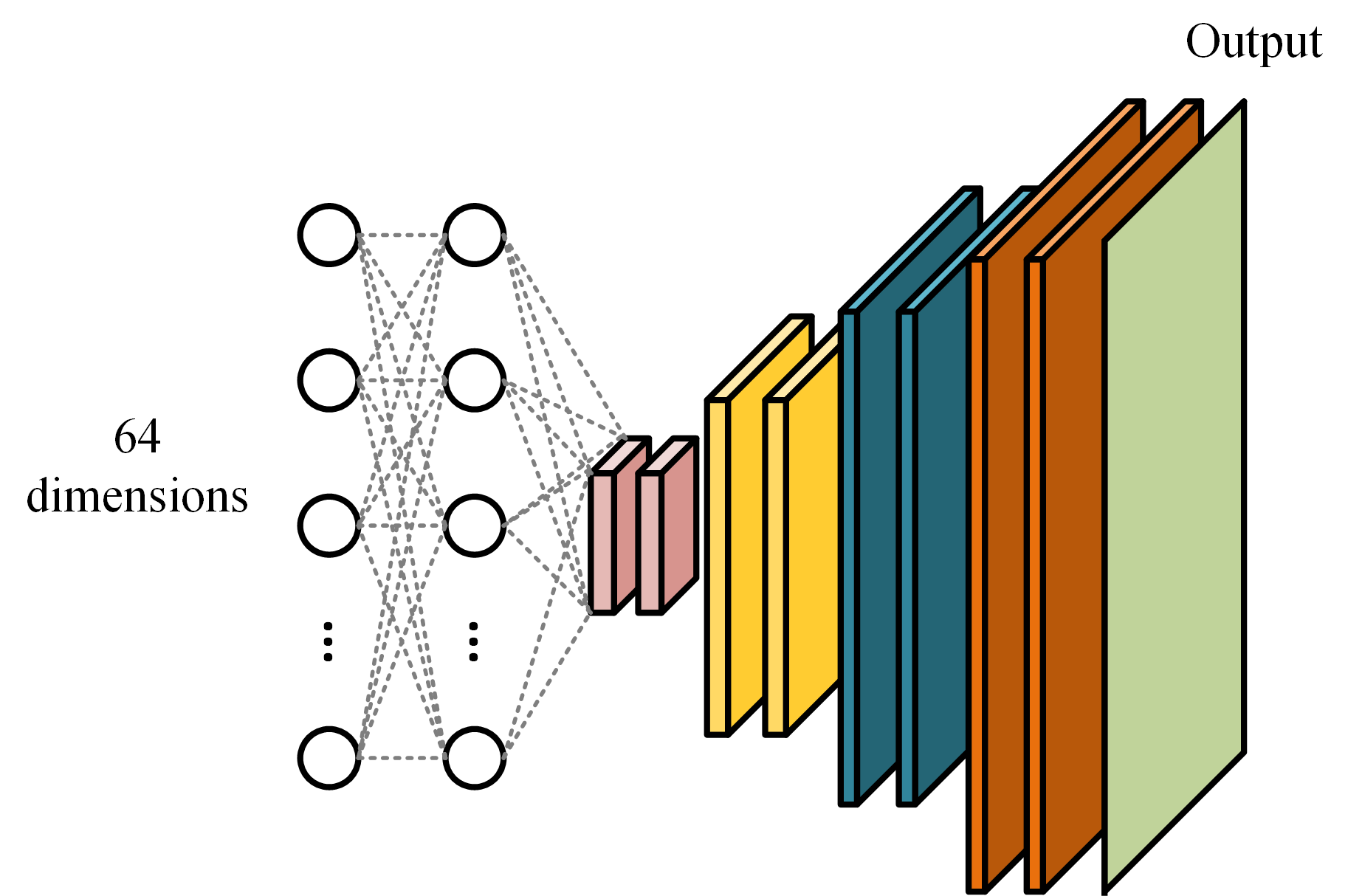}
\caption{Decoder.}
\label{fig:3}  

\end{subfigure}
\caption{Encoder and Decoder architectures used in BEGAN approach.}
\end{figure}

The loss function for training the conditional BEGAN (CBEGAN) is given by:
\begin{equation}
\begin{split}
&L_d = L(x)-k_t\cdot L(G(z|c))\\
&L_g = L(G(z|c))\\
&k_{t+1}=k_t+\lambda_{k}\big(\gamma L(x)-L(G(z|c))\big)
\end{split}
\end{equation}
Where $L_g$ and $L_d$ are generators and discriminators losses respectively. $G$ is the generator function, $z$ is a sample from the latent space, $c$ is the class label, $x$ is the sample drawn from the database, $\lambda_k$ is the learning rate for $k$, $\gamma$ is the equilibrium hyper parameter set to 0.5 in this work, and $L$ is the auto-encoders loss defined by
\begin{equation}
L(v) = |v-D(v)|^2
\end{equation} 
The proposed method needs a classifier to back-propagate the classification error throughout the generator. The classifier used in this experiment is a simple deep classifier given in table \ref{tab:1}.
%
\begin{table}
\caption{the classifier structure for the CelebA+BEGAN experiment.}
\label{tab:1}
\centering       
\begin{tabular}{|l|l|l|l|}
\hline
Layer & Type & kernel & Activation  \\
\hline
Input & Input$(48\times 48)$ & -- & -- \\
\hline
Hidden 1 & Conv & $3\times3$(16 ch)& ReLU \\
\hline
Pool 1 &Max pooling &$2\times2$ &--\\
\hline
Hidden 2 & Conv & $3\times3$(8 ch)& ReLU \\
\hline
Pool 2 &Max pooling &$2\times2$ &--\\
\hline
Hidden 3 & Dense & 1024& ReLU \\
\hline
Output & Dense & 1 & Sigmoid \\
\hline
\end{tabular}
\end{table}
The loss functions used to train the VAC+GAN applied to BEGAN is given by
\begin{equation}
\begin{split}
&L_d = L(x)-k_t\cdot L(G(z|c))\\
&L_g = \vartheta \cdot L(G(z|c)) + \zeta\cdot BCE\\
&k_{t+1}=k_t+\lambda_{k}\big(\gamma L(x)-L(G(z|c))\big)
\end{split}
\end{equation}
where $BCE$ is the binary cross-entropy loss of the classifier, and $\vartheta$ and $\zeta$ are set to 0.997 and 0.003 respectively. The optimizer used for training the generator and discriminator is ADAM with learning rate, $\beta_1$ and $\beta_2$ equal to 0.0001, 0.5 and 0.999 respectively. And the classifier is optimized using nestrov momentum gradient descent with learning rate and momentum equal to 0.01 and 0.9 respectively.\\
The latent space has 64 dimensions and the first dimension is used to partition the latent space in two subspaces corresponding to two classes.
The results for the CBEGAN and proposed method are shown in figures \ref{fig:4n5} and \ref{fig:6n7}.\\
\begin{figure}
\begin{subfigure}[b]{.49\textwidth}
  \includegraphics[width=\columnwidth]{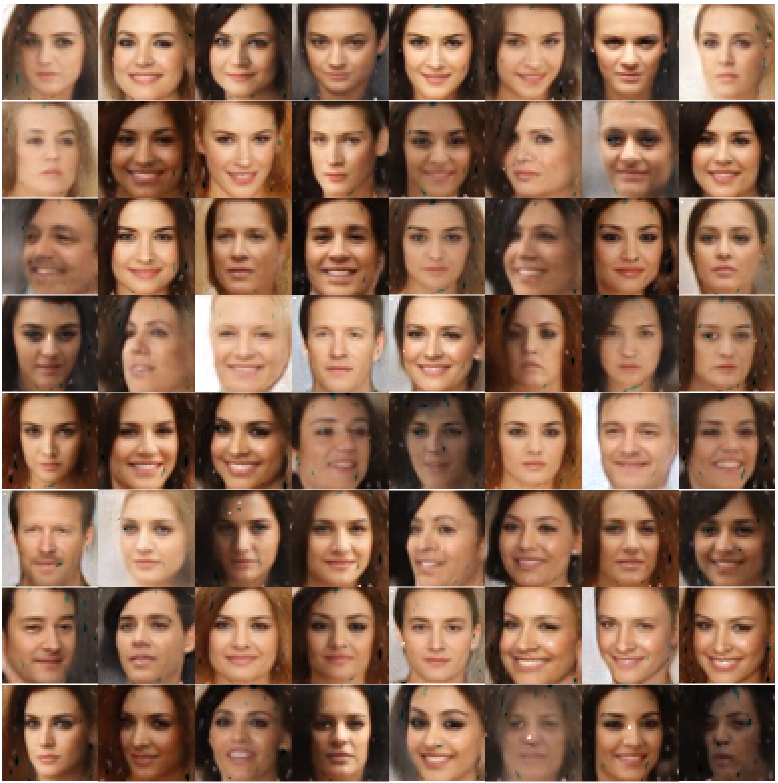}
\caption{Constrained to generate female samples.}
\label{fig:4}       
\end{subfigure}
\begin{subfigure}[b]{.49\textwidth}
  \includegraphics[width=\columnwidth]{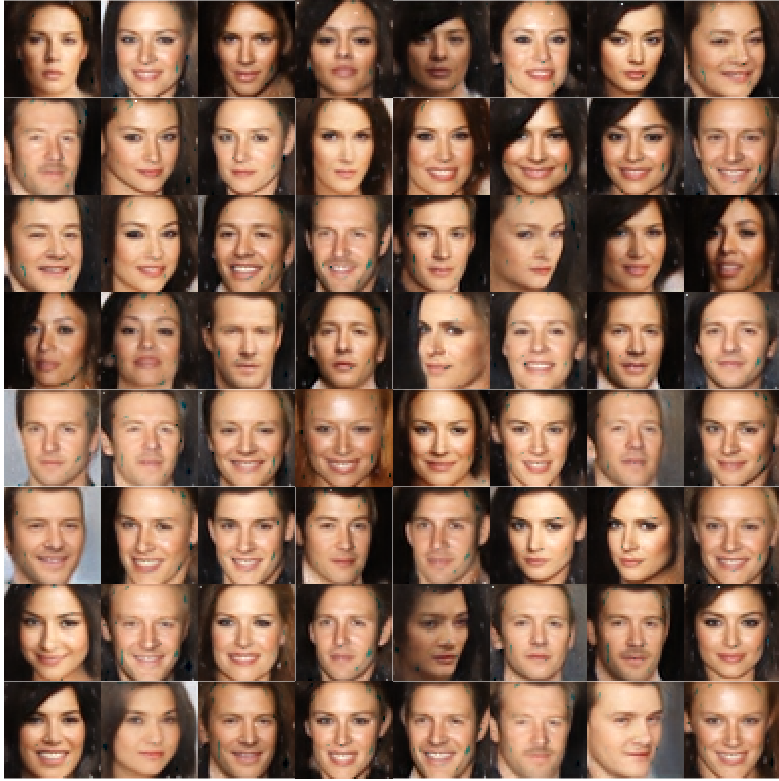}
\caption{Constrained to generate male samples.}
\label{fig:5} 
\end{subfigure}
\caption{Generator trained using CBEGAN method.}
\label{fig:4n5}
\end{figure}
\begin{figure}
\begin{subfigure}[b]{.49\textwidth}
  \includegraphics[width=\columnwidth]{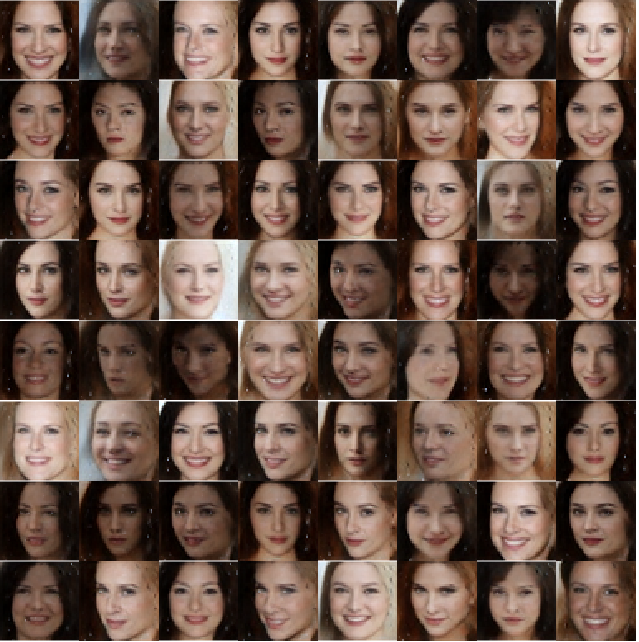}
\caption{Constrained to generate female samples.}
\label{fig:6}       
\end{subfigure}
\begin{subfigure}[b]{.49\textwidth}
  \includegraphics[width=\columnwidth]{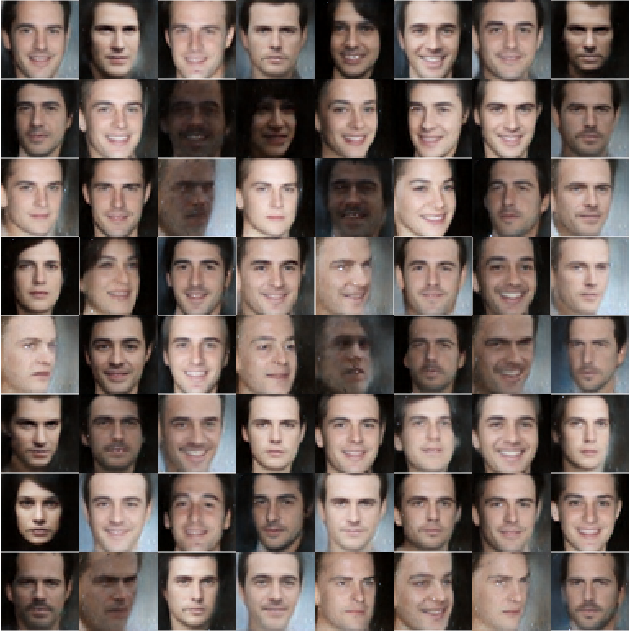}
\caption{Constrained to generate male samples.}
\label{fig:7}
\end{subfigure}
\caption{Generator trained using the proposed method (VAC+GAN).}
\label{fig:6n7}
\end{figure}
As it is shown in these figures, the gender-specific generator fails to correctly generate samples for a specific class when the conditional GAN is applied. But the proposed method is able to correctly constrain the generator to make samples drawn from a specific class.
In order to compare the models, 80 random  male and 80 random female samples have been generated using the trained generators. Three observations have been conducted on these samples:
\begin{enumerate}
\item Each male sample has been compared to all the other male samples, and all the metrics have been calculated for these comparisons, and the average of these numbers has been obtained (blue bars).
\item Each female sample has been compared to all the other female samples, and all the metrics have been calculated for these comparisons, and the average of these numbers has been obtained (purple bars).
\item Each male samples has been compared to all female samples, and all the metrics have been calculated for these comparisons, and the average of these numbers has been obtained (yellow bars).
\end{enumerate}
The aforementioned measurements are illustrated in figures \ref{fig:8} and \ref{fig:9} for the CBEGAN and proposed method.
\begin{figure}
  \includegraphics[width=\columnwidth]{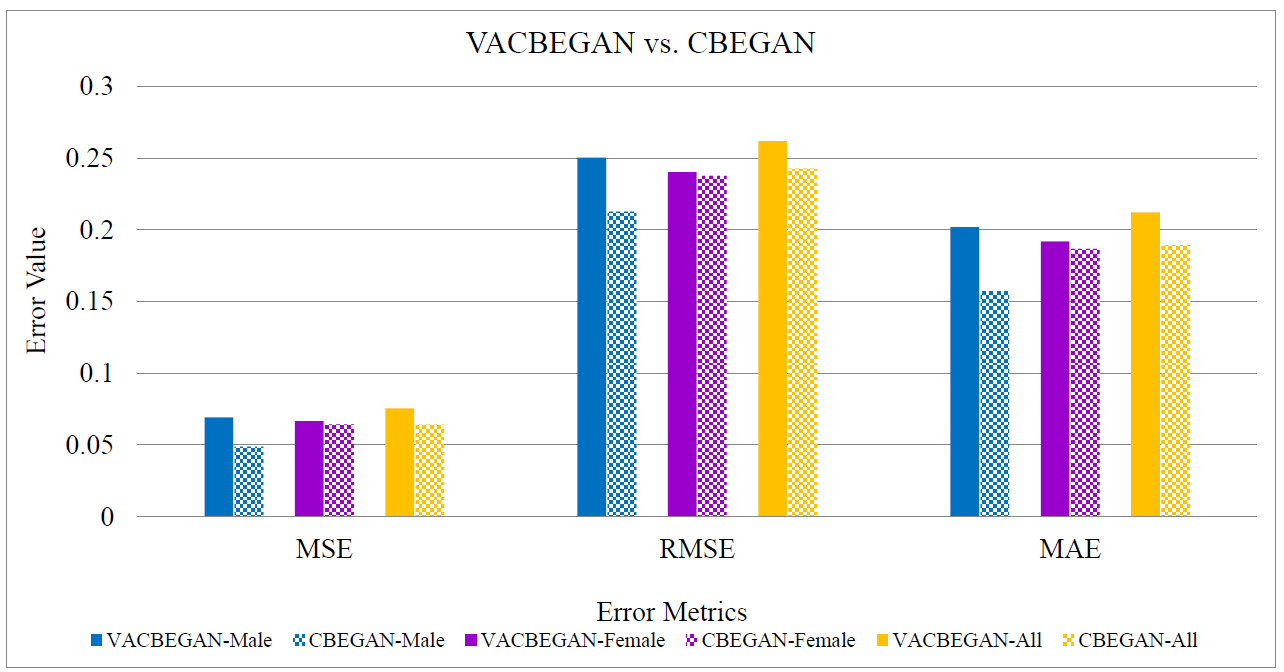}
\caption{UQI  and SSIM metrics for presented method vs. CBEGAN. Lower values show higher performance.}
\label{fig:8}       
\end{figure}
\begin{figure}
  \includegraphics[width=\columnwidth]{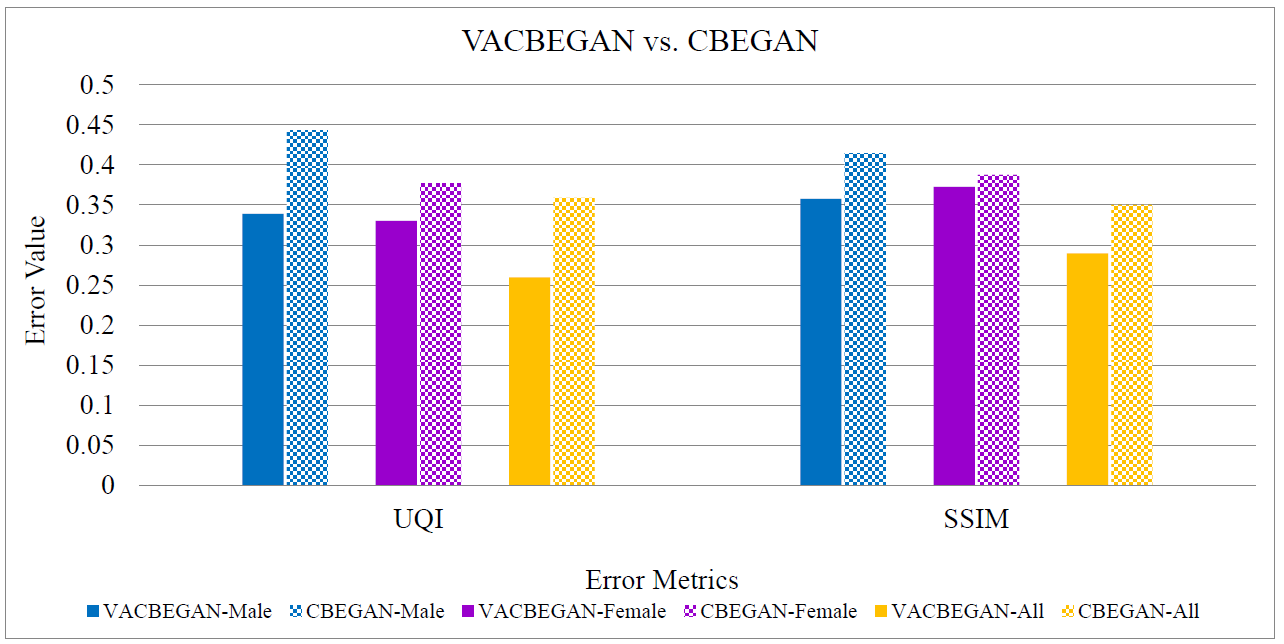}
\caption{MSE, RMSE and MAE metrics for presented method vs. CBEGAN. higher values show better performance.}
\label{fig:9}       
\end{figure}
The lower value of UQI and SSIM shows the less similarity between samples.  In figure \ref{fig:8}, from the first two observations (blue and purple bars) it is shown that the proposed method is able to generate samples in each class that are not similar. From the third observation (yellow bars) it is shown that the inter-class similarity in the proposed method is less than CBEGAN i.e., This shows that the generated samples from different classes are less similar to each other.
The higher value of MSE, RMSE, and MAE shows the higher variation of the generated images. As it is shown in figure \ref{fig:9} the proposed method is able to generate a higher variation of samples for each class and also between classes.
\section{Discussion and Conclusion}
\label{sec:4}
In this work a new approach has been introduced to train conditional deep generators. In this work it has been proven that VAC+GAN is applicable to any GAN framework regardless of the model structure and/or loss function (see Section \ref{sec:2}). The idea is to place a classifier in parallel to the discriminator network and back-propagate the loss of this classifier through the generator network in the training stage. It has also been shown that the presented framework increases the Jensen Shannon Divergence (JSD) between classes generated by the deep generator. i.e., the generator can produce samples drawn from a desired class. The results has been compared to another versatile method known as Conditional GAN (CGAN) for gender specified face generation.\\ 
The future work includes applying the method to datasets with large number of classes and also extend the implementation for bigger size images. Other idea is to extend the current approach to regression problems. This can help to generate samples with specific continuous aspect.
\section*{Acknowledgments}
This research is funded under the SFI Strategic Partnership Program by Science Foundation Ireland (SFI) and FotoNation Ltd. Project ID: 13/SPP/I2868 on Next Generation Imaging for Smartphone and Embedded Platforms.
\bibliographystyle{plain}      
\bibliography{lib} 
\begin{appendices}
\section{Diversity Measurements}
\label{appendixA}
\begin{enumerate}
\item {\bf MSE (Mean Squared Error)}:MSE measures the average of the squares of the errors or deviations; representing the difference between the estimator and what is estimated. The lower value of MSE shows lesser error. 
\begin{equation}
MSE(f,g) = \frac{1}{mn}\sum_{0}^{m-1}\sum_{0}^{n-1}||f(i,j)-g(i,j)||
\end{equation}
\item {\bf RMSE (Root Mean Squared Error)}: RMSE is a quadratic scoring rule that measures the average magnitude of the error. It is the square root of the average of squared differences between prediction and actual observation. The lower value of RMSE shows lesser error.
\begin{equation}
RMSE(y,\hat{y}) = \sqrt{\frac{1}{n}\sum_{i=1}^{n}(y_i-\hat{y_i})^2}
\end{equation}
\item {\bf MAE (Mean Absolute Error)}:MAE also measures the average magnitude of the errors in a set of predictions, without considering their direction. It is the average over the test sample of the absolute differences between prediction and actual observation where all individual differences have equal weight. The lower value of MAE shows lesser error.
\begin{equation}
MAE(f,y) = \frac{1}{n}\sum_{i=1}^{n}|f_i-y_i|
\end{equation}
\item {\bf UQI (Universal Quality Index)} \cite{UQI}: UQI measures the structural distortion of the images by modeling the distortion as a combination of three factors: loss of correlation, luminance distortion, and contrast distortion. The higher value of UQI shows lesser error.
\item {\bf SSIM (Structural Similarity Index)} \cite{SSIM}: SSIM is a perception-based model that considers image degradation as perceived change in structural information, while also incorporating important perceptual phenomena, including both luminance masking and contrast masking terms. The higher value of SSIM shows lesser error.
\end{enumerate}
\end{appendices}

\end{document}